\documentclass[letterpaper, 10 pt, conference]{ieeeconf}  
\pdfoutput=1

\usepackage{amsmath} 
\usepackage{amssymb}  
\usepackage{graphicx}

\usepackage{amsthm}
\allowdisplaybreaks

\IEEEoverridecommandlockouts  

\title{\LARGE \bf
Approximate Optimal Control for Safety-Critical Systems 
\\with Control Barrier Functions
}

\author{Max H. Cohen and Calin Belta 
\thanks{This material is based upon work supported by the National Science Foundation Graduate Research Fellowship Program under Grant No. DGE-1840990. Any opinions, findings, and conclusions or recommendations expressed in this material are those of the author(s) and do not necessarily reflect the views of the National Science Foundation.}
\thanks{The authors are with the Department of Mechanical Engineering, 
        Boston University, 110 Cummington Mall, Boston, MA 02215, United States
        {\tt\small \{maxcohen, cbelta\}@bu.edu}}%
}

\newtheorem{theorem}{Theorem}

\newtheorem{proposition}{Proposition}

\theoremstyle{definition}
\newtheorem{definition}{Definition}

\theoremstyle{definition}
\newtheorem{problem}{Problem}

\theoremstyle{definition}
\newtheorem{assumption}{Assumption}

\theoremstyle{remark}

\begin{document}

\maketitle
\thispagestyle{empty}
\pagestyle{empty}

\begin{abstract}

  Control Barrier Functions (CBFs) have become a popular tool for enforcing set invariance in safety-critical control systems. While guaranteeing safety, most CBF approaches are myopic in the sense that they solve an optimization problem at each time step rather than over a long time horizon. This approach may allow a system to get too close to the unsafe set where the optimization problem can become infeasible. Some of these issues can be mitigated by introducing relaxation variables into the optimization problem; however, this compromises convergence to the desired equilibrium point. To address these challenges, we develop an approximate optimal approach to the safety-critical control problem in which the cost of violating safety constraints is directly embedded within the value function. We show that our method is capable of guaranteeing both safety and convergence to a desired equilibrium. Finally, we compare the performance of our method with that of the traditional quadratic programming approach through numerical examples.
\end{abstract}

\section{INTRODUCTION}\label{sec:intro}
The concept of safety has received much attention in the fields of robotics and controls over the past few years. One of the prime reasons for this is the rise of autonomy for safety-critical systems such as self-driving cars. This has led to the question: how does one formally define what it means to be safe? Informally speaking, one could define safety as something “bad” never happens; however, more formal definitions of safety have been linked to the concept of set invariance \cite{Blanchini99}. A popular technique for enforcing set invariance in safety-critical systems is the Control Barrier Function (CBF) approach \cite{CBF19,AmesTAC17}. These methods typically involve synthesizing a safe controller by embedding set invariance conditions within an optimal control problem. Rather than solving a general constrained optimal control problem however, most papers propose to discretize time and assume a piecewise constant control. If the control system is affine in controls and the cost is quadratic, the problem reduces to solving a quadratic program (QP) at each time step to obtain the optimal control \cite{AmesTAC17,HOCBF}. 

One issue with the QP-based approach is that it operates myopically, that is, the safe control is only a function of the current state \cite{AmesDensityFuncOpt}. While this approach can guarantee local safety at each time step, the satisfaction of the safety constraint is dependent on how frequently the QP is solved \cite{GuangCBF}. A step size too small can induce unnecessary computation whereas a step size too large can result in unsafe behavior. Additionally, the QP may allow trajectories to approach the boundary of the safe set very closely before intervening. Consequently, when the system approaches the boundary of the safe set the QP may become infeasible and the approach fails \cite{CBF_ML}. The feasibility of the QP can be increased by introducing relaxation variables; however, this compromises convergence to the desired equilibrium point which may no longer be guaranteed \cite{JankovicAutomatica18}. Moreover, one must take care when simply merging stabilizing conditions with safety conditions as this can shift the desired equilibrium point of the closed-loop system \cite{CLF_CBF_unstable}.

If the CBF problem is not framed in terms of a QP then one is faced with the task of solving a general constrained optimal control problem. One way to obtain a solution to an optimal control problem is to solve the Hamilton-Jacobi-Bellman (HJB) equation; however, for many systems this involves solving a nonlinear partial differential equation (PDE) which typically does not have a closed-form solution \cite{Liberzon_Opt}. The approach commonly taken is to numerically solve the HJB equation offline to generate a control policy which is then implemented on the system in real time. Along these lines, recent work has proposed using density functions, which are the dual to the value function in optimal control, to enforce safety \cite{AmesDensityDual}. It was shown in \cite{AmesDensityFuncOpt} that CBF constraints can be embedded within the density function and the resulting optimal control problem can be solved with a primal-dual algorithm. This approach addresses the myopic nature of the QP method; however, the solution is obtained by discretizing the state space and solving the HJB PDE offline which is computationally demanding. Other authors proposed using neural networks (NNs) to learn safe control policies subject to CBF constraints \cite{Jyo_CBF}; however, these results don't present stability guarantees and the solution is obtained offline. One issue with offline solutions is that they can become computationally demanding as the complexity of the system increases. Additionally, offline solutions are poorly-suited for safety-critical tasks as they are not robust to uncertainties in the system and environment. Therefore, there is a need for online solutions to the safety-critical optimal control problem. 

Recently, reinforcement learning (RL) inspired methods such as approximate dynamic programming (ADP) have been proposed to approximately solve optimal control problems online (see \cite{ADP_survey1,ADP_survey2} for a survey). These methods utilize an actor-critic structure where the critic learns the optimal value function and the actor learns the optimal control input. This method was used to solve infinite-horizon optimal regulation problems for nonlinear continuous-time systems online in \cite{Vam10} and more recent work has focused on various extensions \cite{Bhasin13,Rushi15,Vam15,RUSHi16,ADP_StaF,ADP_PathPlanning18,Vam_ACC19}.  In most literature the actor and critic are parameratized as NNs and although the solution is obtained online, the computational demands of the NNs may inhibit real-time implementation on physical systems. Because of this, other works have focused on developing computationally efficient approximation methods which are able to approximate functions in a local neighborhood of the current state \cite{ADP_StaF}. These computationally efficient ADP methods have been successfully used in some safety-critical applications such as robot motion planning \cite{ADP_PathPlanning18}; however, designing provably safe ADP controllers for general safety-critical systems is still an open area of research \cite{Vam_ACC19}.

In this paper we present an ADP method to solve the safety-critical optimal control problem online in which safety-invariants are expressed as barrier functions. In Sec. \ref{sec:prelim} we introduce formal notions of safety used in the current literature and formulate the general problem under consideration. In Sec. \ref{sec:reform} we reformulate the traditional problem as an unconstrained optimal control problem and show that the solution to this new problem guarantees satisfaction of the original constraints. Sec. \ref{sec:ADP} provides an ADP solution to the reformulated problem from Sec. \ref{sec:reform} and Sec. \ref{sec:analysis} presents a Lyapunov-based analysis in which the ADP method is shown to guarantee both convergence and safety of this solution. Finally, we provide numerical examples in Sec. \ref{sec:sim} and finish with concluding remarks in Sec. \ref{sec:conclusion}. In comparison to the current QP approach our method: 1) shows improved convergence to a stable equilibrium, 2) has increased feasibility, and 3) is not dependent on discretizing the time.  To the best of our knowledge this is also the first attempt to use CBFs to design provably safe ADP controllers.


\section{PRELIMINARIES AND PROBLEM FORMULATION}\label{sec:prelim}

Throughout this paper we consider affine control systems of the form
\begin{equation}
  \label{eq:dyn}
  \dot{x}(t)=f(x(t))+g(x(t))u(t),\quad x(0)=x_0,
\end{equation}
where $x(t)\in\mathbb{R}^{n}$ denotes the system state, $f\,:\,\mathbb{R}^{n}\rightarrow\mathbb{R}^{n}$ models the system drift, the columns of $g\,:\,\mathbb{R}^{n}\rightarrow\mathbb{R}^{n\times m}$ capture the control directions, $u(t)\in U\subset\mathbb{R}^{m}$ is the control input, and $U$ denotes the control constraint set. Note that the explicit dependence on time will be dropped unless needed for clarity. We assume the functions $f,\,g$ are locally Lipschitz continuous, $f(0)=0$, $f$ is sufficiently smooth and $0<||g(x)||\leq\bar{g}$ with $\bar{g}\in\mathbb{R}_{>0}$ where $||\cdot||$ denotes the 2-norm. 
To formalize the concept of safety we introduce the following:

\begin{definition}[Forward Invariance]
  \label{def:Forward-Invariance}
  Consider a set $C\subseteq\mathbb{R}^{n}$ and initial condition $x(0)=x_{0}$. The set $C$ is \emph{forward invariant} for system \eqref{eq:dyn} if $x_{0}\in C\implies x(t)\in C,\,\forall t\geq0$.
\end{definition}

In the current literature \cite{CBF19,AmesTAC17,HOCBF}, if a set $C$ can be rendered forward invariant, then system \eqref{eq:dyn} is said to be \emph{safe} with respect to $C$. In this paper, we assume that the set\footnote{For a set $C$, the notation $\partial C$ denotes the boundary of $C$ and $\text{Int}(C)$ denotes its interior.} $C$ is described by the superlevel set of a continuously differentiable function $h\,:\,\mathbb{R}^n\rightarrow\mathbb{R}$ \cite{CBF19} such that 
  \begin{subequations}
      \label{eq:safe_set}
      \begin{align}
          C = \{x\in\mathbb{R}^n\,|\,h(x)\geq0\},\\
          \partial C = \{x\in\mathbb{R}^n\,|\,h(x)=0\},\\
          \text{Int}(C) = \{x\in\mathbb{R}^n\,|\,h(x)>0\}.
      \end{align}
  \end{subequations}
  
\begin{definition}[Control barrier function \cite{HOCBF}]
  \label{def:CBF}
  The function $h$ in \eqref{eq:safe_set} is a \emph{control barrier function} (CBF) for system \eqref{eq:dyn} if there exists a class $\mathcal{K}$ function $\alpha$ such that
  \begin{equation} \label{eq:CBF}
      L_{f}h(x)+L_{g}h(x)u+\alpha\left(h(x)\right)\geq0,\quad\forall x\in C,
  \end{equation}
  where $L_{f}h(x)=\frac{dh}{dx}f(x)$ is the Lie derivative of $h$ along $f$.
\end{definition}

\begin{theorem}[\cite{CBF19}]\label{thm:CBF}
Let $C$ be defined as in \eqref{eq:safe_set}. If $h$ is a CBF on $C$ and $\frac{\partial h}{\partial x}(x)\neq0,\,\forall x\in\partial C$ then any Lipschitz continuous controller $u(x)\in K_{cbf}(x)$ for \eqref{eq:dyn}, where
\begin{equation}\label{eq:k_cbf}
    K_{cbf}(x)\triangleq\{u\in U\,|\,L_{f}h(x)+L_{g}h(x)u+\alpha\left(h(x)\right)\geq0\},
\end{equation}
renders $C$ forward invariant.
\end{theorem}

The above theorem illustrates that the existence of a CBF implies the safety of \eqref{eq:dyn}. However, given certain assumptions on $C$, it has been shown that CBFs provide necessary and sufficient conditions for safety, which is formalized through the following theorem:

\begin{theorem}[\cite{CBF19}] \label{thm:FI_iff_CBF}
Let $C$ be a compact set defined by \eqref{eq:safe_set} with the property that $\frac{\partial h}{\partial x}(x)\neq 0,\,\,\forall x\in\partial C$. If there exists a control law $u$ that renders $C$ forward invariant, then $h\,:\,C\rightarrow\mathbb{R}$ is a CBF on $C$.
\end{theorem}

\begin{definition}[Control Lyapunov Function \cite{CBF19}]
\label{def:CLF}
A continuously differentiable function $V_{clf}\,:\,\mathbb{R}^n\rightarrow\mathbb{R}_{\geq0}$ is a \emph{control Lyapunov function} (CLF) for \eqref{eq:dyn} if it is positive definite and satisfies 
\begin{equation}
    \label{eq:CLF2}
    \underset{u\in U}{\inf}\left[L_fV_{clf}(x)+L_gV_{clf}(x)u+\gamma(V_{clf}(x))\right]\leq0, 
\end{equation}
where $\gamma\,:\mathbb{R}_{\geq0}\rightarrow\mathbb{R}_{\geq0}$ is a class $\mathcal{K}$ function.
\end{definition}

\begin{theorem}[\cite{CBF19}]\label{thm:CLF}
Given system \eqref{eq:dyn}, if there exists a CLF $V_{clf}(x)\geq0$ satisfying \eqref{eq:CLF2}, then any Lipschitz continuous feedback controller $u(x)\in K_{clf}(x)$ where 
\begin{multline}
    \label{eq:CLF3}
    K_{clf}(x)\triangleq\{u\in U\,|\,L_fV_{clf}(x)\\+L_gV_{clf}(x)u+\gamma(V_{clf}(x))\leq0\},
\end{multline}
asymptotically stabilizes the system to $x=0$.
\end{theorem}

Now consider the cost functional 

\begin{equation}\label{eq:J}
    J(x,u)\triangleq\int_{0}^{\infty}{r(x(\tau),u(\tau))}d\tau,
\end{equation}

where $r\,:\,\mathbb{R}^n\times\mathbb{R}^m\rightarrow\mathbb{R}_{\geq0}$ is an instantaneous positive definite cost. Consider the following problem:

\begin{problem}
\label{problem:main}
Consider system \eqref{eq:dyn} with initial condition $x_0\in \text{Int}(C)$. Find a control $u\in U$ that drives the system from $x_0$ to the origin, while minimizing \eqref{eq:J} and keeping the system safe. 
\end{problem} 

To solve Problem \ref{problem:main} existing works \cite{AmesTAC17,CBF19,HOCBF} propose to view \eqref{eq:k_cbf} and \eqref{eq:CLF3} as constraints in an optimal control problem. Time is then discretized and the system state is assumed to be fixed at the start of each time interval. Consequently the constraints become linear in the control and, if $r$ is quadratic in $u$, the problem reduces to solving a QP at each time step. This constant control is then applied to the continuous system \eqref{eq:dyn} over the entire time interval and the procedure is repeated at each time step. Specifically, the QP solved is of the form
\begin{subequations}
    \label{eq:CBF_QP_Sim}
    \begin{align}
      \min_{u\in U}\quad & u^{T}Ru+p\varphi^2\\
      \text{s.t.}\quad &   L_{f}h(x)+L_{g}h(x)u +\alpha(h(x)) \geq0,\\
                 \quad & L_{f}V_{clf}(x)+L_{g}V_{clf}(x)u\leq-\gamma (V_{clf}(x))+\varphi,
    \end{align}
  \end{subequations}
  where $\varphi\in\mathbb{R}$ is a relaxation variable which is penalized by $p\in\mathbb{R}_{>0}$, $R\in\mathbb{R}^{m\times m}$ is the control penalty, and $\alpha,\gamma$ are the class $\mathcal{K}$ functions from \eqref{eq:CBF} and \eqref{eq:CLF2}, respectively. The relaxation variable is added to increase the feasibility of the QP which can easily become infeasible in the presence of conflicting control, stability, and safety constraints \cite{CBF_ML}. While increasing feasibility, this relaxation no longer guarantees convergence to the desired equilibrium point \cite{JankovicAutomatica18}. To address these issues we seek a solution to Problem \ref{problem:main} by formulating an optimal control problem whose solution satisfies the control, stability, and safety constraints without relying on the discretization of time. To this end we propose to augment the instantaneous cost $r$ with additional terms whose minimization imply satisfaction of the original constraints.

\section{PROBLEM REFORMULATION AND APPROACH}\label{sec:reform}
Consider Problem \ref{problem:main} with the cost functional in \eqref{eq:J}. Rather than dealing with a constrained problem we seek to reformulate Problem \ref{problem:main} as an unconstrained optimal control problem. To this end we redefine the instantaneous cost as

\begin{equation}
\label{eq:r}
    r(x,\,u)\triangleq x^TQx + R_u(u) + B(x),
\end{equation}
where $Q\in\mathbb{R}^{n\times n}$ is a positive definite matrix which penalizes the state, $R_u\,:\,\mathbb{R}^{m}\rightarrow\mathbb{R}_{\geq0}$ is a positive definite function which penalizes and ensures boundness of the control, and $B\,:\,\text{Int}(C)\rightarrow\mathbb{R}_{\geq0}$ is a barrier-like function that satisfies
\begin{equation}
  \label{eq:RBF}
      \begin{aligned}
      \underset{x\in\text{Int}\left(C\right)}{\inf} B(x)\geq0, \quad & \underset{x\rightarrow\partial C}{\lim}B(x)=\infty, \quad & B(0)=0.
      \end{aligned}
  \end{equation} 
Based on \eqref{eq:safe_set}, a choice of $B$ which satisfies \eqref{eq:RBF} is $B(x)=\frac{s(x)}{h(x)}$ where $s\,:\,\mathbb{R}^n\rightarrow[0,1]$ is a user-defined smooth scheduling function\footnote{It is assumed that the smooth scheduling function is designed such that $s(0)=0$. See \cite{ADP_PathPlanning18} for examples of scheduling functions.} that ensures trajectories are only penalized near $\partial C$. The state penalty matrix $Q$ from \eqref{eq:r} is positive defnite and hence satisfies $\underline{q}\|x\|^2\leq x^TQx\leq\overline{q}\|x\|^2$ with $\underline{q},\,\overline{q}\in\mathbb{R}_{>0}$ for all $x\in\mathbb{R}^n$. Moreover, we assume the control constraint set $U$ is defined by symmetric input constraints such that $U=\{ u\in\mathbb{R}^m\,|\,-\overline{u}\leq u_i\leq \overline{u},\,i=1,...,m\}$, where $u_i$ is the $i$th component of $u$ and $\overline{u}\in\mathbb{R}_{>0}$ is the maximum allowable control. A popular approach to enforcing such control constraints is to use a non-quadratic control cost of the form \cite{Vam15,Bounded_u_pen}
\begin{equation}
\label{eq:Ru}
    R_u(u)\triangleq2\sum_{i=1}^{m}{\int_{0}^{u_i}{\bar{u}r_{i}\tanh^{-1}\left(\zeta_{i}/\bar{u}\right)}d\zeta_{i}},
\end{equation}
where $r_i\in\mathbb{R}_{>0}$ are components that form a diagonal positive definite matrix $R\in\mathbb{R}^{m\times m}$ as $R\triangleq\text{diag}\{\bar{r}\}$ and $\bar{r}\triangleq\left[r_1,\dots,r_m\right]^T$. If the time-horizon is infinite and the system and cost are time-invariant then the optimal value function $V^*\,:\,\mathbb{R}^n\rightarrow\mathbb{R}_{\geq0}$ is also time-invariant and can be expressed as $V^*(x)=\underset{u(\tau)\in U}{\inf}\int_{t}^{\infty}{r(x(\tau),u(\tau))d\tau}$. The associated Hamiltonian is $H\left(x,u,\nabla V^*\right)=L_{f}V^*(x)+L_{g}V^*(x)u+r(x,u)$
which can be used with the stationary condition $\partial H/\partial u=0$ to derive the optimal controller as

\begin{equation}
\label{eq:u*}
    u^*(x)=-\bar{u}\text{Tanh}\left(\frac{R^{-1}g(x)^T}{2\bar{u}}\nabla V^{*}(x)^T\right),
\end{equation}
where $\nabla(\cdot)$ denotes the derivative of $(\cdot)$ with respect to its first argument and $\text{Tanh}(\zeta)\triangleq[\tanh(\zeta_{i}),...,\tanh(\zeta_{m})]^{T},\,\forall\zeta\in\mathbb{R}^m$. 
The optimal value function and controller satisfy the HJB equation
\begin{equation}
\label{eq:HJB}
    0=\underset{u\in U}{\min}\,H=L_{f}V^{*}(x)+L_{g}V^{*}(x)u^{*}+r\left(x,u^*\right),
\end{equation}
with a boundary condition of $V^*(0)=0$. 

\begin{proposition}\label{prop:1}
Let the origin be contained in $C$ and let $x_0\in \text{Int}(C)$. Further, assume that there exists a smooth function $V^*(x)\geq0$ which satisfies \eqref{eq:HJB}. Then, the closed-loop system composed of \eqref{eq:dyn} and controller \eqref{eq:u*} solves Problem \ref{problem:main}.
\end{proposition}
\begin{proof}
By definition $V^*$ is positive definite and satisfies $V^*(0)=0$, making it a suitable Lyapunov function candidate. Taking the derivative of $V^*$ along the trajectories of \eqref{eq:dyn} yields
\begin{equation}\label{eq:V*_dot}
        \dot{V}^*(x)=L_fV^*(x)+L_gV^*(x)u^*(x)\leq-\underline{q}\left\Vert x\right\Vert^2,
\end{equation}
where $L_fV^*(x)=-L_gV^*(x)u^*(x)-r(x,u^*)$ from \eqref{eq:HJB} and $-x^TQx-R_u(u^*)-B(x)\leq\underline{q}\left\Vert x\right\Vert^2$ from \eqref{eq:RBF}, \eqref{eq:Ru} were used. Since $V^*$ was used as a Lyapunov function candidate, it follows from \eqref{eq:V*_dot} and \cite[Theorem 4.1]{Khalil} that the origin is asymptotically stable for \eqref{eq:dyn}. Additionally, $u^*$ maps from $\mathbb{R}^n\rightarrow(-\bar{u},\bar{u})$ thus, the input constraints are satisfied. Now suppose $x_0\in C$ and $C$ is not forward invariant. Then $\exists\, \bar{t}\geq0$ such that $x(\bar{t})\rightarrow\partial C\implies B(x(\bar{t}))\rightarrow\infty\implies V^*(x(\bar{t}))\rightarrow\infty$ which contradicts \eqref{eq:V*_dot}. Thus, $\nexists\,t\geq0$ for which $x\rightarrow\partial C$ so $x_0\in C\implies x\in C,\,\forall t\geq0$ and by Def. \ref{def:Forward-Invariance} $C$ is forward invariant and \eqref{eq:dyn} is safe. Moreover, if $C$ is compact it follows from Theorem \ref{thm:FI_iff_CBF} that $h\,:\,C\rightarrow\mathbb{R}$ is a CBF for \eqref{eq:dyn} over $C$ and $u^*(x)\in K_{cbf}(x)$.
\end{proof}

Proposition \ref{prop:1} illustrates that the solution to the unconstrained infinite-horizon optimal control problem with a cost defined by \eqref{eq:r} solves Problem \ref{problem:main}; however, this is conditioned on solving the HJB equation \eqref{eq:HJB} for $V^*$. Generally speaking, \eqref{eq:HJB} is a nonlinear PDE which cannot be solved analytically. To address this issue, we propose an ADP approach in which the optimal value function is learned online.

\section{APPROXIMATE DYNAMIC PROGRAMMING}\label{sec:ADP}
In the following, we develop a local approximation scheme and online update laws to learn the solution to the HJB equation online. 

\subsection{Value Function Approximation}
Consider the compact set $\chi\subset\mathbb{R}^{n}$ with $x$ in the interior of $\chi$ and let $\Omega(x)$ denote a small compact set centered at the current state $x$. The value function can be represented at points $y\in\Omega(x)$ using state following (StaF) kernels \cite{ADP_StaF,StaF} as
\begin{equation}
    \label{eq:StaF_V_opt}
    V^{*}(y)=W(x)^{T}\sigma\left(y,\,c(x)\right)+\epsilon(x,y),
\end{equation}
where $W\,:\,\chi\rightarrow\mathbb{R}^{L}$ is the continuously differentiable ideal weight function, $\sigma\,:\,\chi\times\chi\rightarrow\mathbb{R}^{L}$ is a vector of $L\in\mathbb{N}$ continuously differentiable bounded positive definite kernel functions, and $c_i(x)\in\chi,\,i=1,\dots,L,$ are the distinct centers of each kernel. The function $\epsilon\,:\,\chi\times\chi\rightarrow\mathbb{R}$ is the function approximation reconstruction error which is assumed to be bounded over $\chi$. Adding and subtracting a bounded version of the barrier-like function \eqref{eq:RBF}, denoted as $\bar{B}\,:\,\text{Int}(C)\rightarrow\mathbb{R}$ where $C$ is the set defined in \eqref{eq:safe_set}, from \eqref{eq:StaF_V_opt}, taking the gradient, and substituting into \eqref{eq:u*} yields an expression for the optimal policy as 
\begin{equation}
    \label{eq:StaF_u_opt}
    u^{*}(y)=-\bar{u}\text{Tanh}\left(\frac{R^{-1}g(y)^T}{2\bar{u}}D^*(y)\right),
\end{equation}
where $D^*(y)\triangleq\nabla\sigma\left(y,\,c\left(x\right)\right)^{T}W\left(x\right) +\nabla W(x)^{T}\sigma(y,\,c(x))+\nabla\epsilon\left(x,y\right)^{T} + \nabla\bar{B}(y)^T$. The addition and subtraction of $\bar{B}$ is made to facilitate the analysis in Sec. \ref{sec:analysis}. If $B$ is chosen as $B(x)=\frac{s(x)}{h(x)}$ then $\bar{B}$ can always be constructed as $\bar{B}(x)=\frac{s(x)}{h(x)+a}$ where $a\in\mathbb{R}_{>0}$ is a positive constant.

In general, the ideal weight function $W$ is unknown a priori and must be replaced with an estimated weight function $\hat{W}(t)\in\mathbb{R}^L$. Similar to most ADP approaches, we maintain separate weight estimates for the value function and optimal policy, denoted as $\hat{W}_c(t),\hat{W}_a(t)\in\mathbb{R}^L$, respectively. Using these estimated weights in the StaF parameterizations of the value function \eqref{eq:StaF_V_opt} and optimal policy \eqref{eq:StaF_u_opt} results in the approximate value function
\begin{equation}
    \label{eq:V_hat}
    \hat{V}(y,x,\hat{W}_{c})\triangleq\hat{W}_{c}^{T}\sigma(y,\,c(x)) + \bar{B}(y),
\end{equation}
and approximate optimal policy
\begin{equation}
    \label{eq:u_hat}
    \hat{u}(y,x,\hat{W}_{a})\triangleq-\bar{u}\text{Tanh}\left(\frac{R^{-1}g(y)^T}{2\bar{u}}\hat{D}(y,x,\hat{W}_a)\right),
\end{equation}
where $\hat{D}(y,x,\hat{W}_a)\triangleq\left(\nabla\sigma(y,\,c(x))^{T}\hat{W}_{a} + \nabla\bar{B}(y)^T\right)$. The notation $\hat{V}(y,x,\hat{W}_c)$ denotes the approximate value function evaluated at $y$, using a kernel centered at $x$, with a weight estimate of $\hat{W}_c$. The expressions for the approximate optimal value function and policy in \eqref{eq:V_hat} and \eqref{eq:u_hat} can then be substituted into \eqref{eq:HJB} to obtain an expression for the approximate HJB equation as 
 \begin{multline}
 \label{eq:BE}
     \hat{H}(y,x,\hat{W}_c,\hat{W}_a)=r(y,\hat{u}(y,x,\hat{W}_a))+L_{f}\hat{V}(y,x,\hat{W}_c)\\+L_{g}\hat{V}(y,x,\hat{W}_c)\hat{u}(y,x,\hat{W}_a),
 \end{multline}
where $\hat{H}\,:\,\mathbb{R}^n\times\mathbb{R}^n\times\mathbb{R}^L\times\mathbb{R}^L\rightarrow\mathbb{R}$ is the approximate Hamiltonian. Taking the difference between the approximate and optimal Hamiltonian as $\delta(y,x,\hat{W}_c,\hat{W}_a)\triangleq\hat{H}(y,x,\hat{W}_c,\hat{W}_a)-H\left(x,u^*,\nabla V^*\right)$ yields the residual approximation error $\delta$, referred to as the Bellman error (BE). From \eqref{eq:HJB}, $H\left(x,u^*,\nabla V^*\right)=0$, thus the BE is just the approximate Hamiltonian. From Proposition \ref{prop:1}, if $\hat{V}\rightarrow V^* $ and $\hat{u}\rightarrow u^*$, then implementing $\hat{u}$ on \eqref{eq:dyn} will solve Problem \ref{problem:main}. Thus, we are faced with the problem of developing estimates of the ideal weights $\hat{W}_c,\hat{W}_a$ that minimize the BE.

\subsection{Online Learning}
In this section we develop online update laws for the estimated weights that ensure convergence to their ideal values. In traditional ADP approaches \cite{Vam10,Bhasin13} a persistence of excitation (PE) condition is required to ensure convergence of the weight estimates; however, this typically involves adding an exploration signal into the system. In addition to degrading performance, the introduction of an exploration signal could compromise safety. More recent works \cite{RUSHi16} have leveraged techniques from concurrent learning adaptive control \cite{Chowdhary} in the form of BE extrapolation which allows the BE to be evaluated at unexplored regions of the state-space. This extrapolation results in a virtual excitation of the system which facilitates weight estimate convergence \cite{RUSHi16}.  To this end, at each time step the BE is extrapolated to a set of points $\{x_{k}(t)\in\Omega(x(t))\,|\,k=1,...,N\}$ about the current state $x(t)$. In the following, let $\delta(t)\triangleq\delta(x(t),x(t),\hat{W}_c(t),\hat{W}_a(t))$ and let the subscript $k$ denote that a function is evaluated at the extrapolated state $x_k(t)$, i.e. $\delta_k(t)\triangleq\delta(x_k(t),x(t),\hat{W}_c(t),\hat{W}_a(t))$. Additionally, let the control $u(t)\triangleq\hat{u}(x(t),x(t),\hat{W}_a(t))$ be the input that drives \eqref{eq:dyn}. For notational brevity, the BE can be expressed more compactly as $\delta(t)=\hat{W}_{c}(t)^{T}\omega(t)+r(x(t),u(t)) + \omega_{B}(t)$
where $\omega(t)\triangleq\nabla\sigma(x(t),c(x(t))(f(x(t))+g(x(t))u(t)$, $\omega_{B}(t)\triangleq \nabla\bar{B}(x(t))(f(x(t))+g(x(t))u(t))$.
To derive an update law for $\hat{W}_c$ consider a squared, normalized version of the BE as $E(t)\triangleq\frac{1}{2}\left(\frac{k_{c1}\delta^{2}(t)}{\rho^{2}(t)} + \sum_{k=1}^{N}{\frac{k_{c2}\delta_{k}^{2}(t)}{N\rho_{k}^{2}(t)}}\right)$ where $k_{c1},\,k_{c2}\in\mathbb{R}_{>0}$ are gains and $\rho$ is a normalization term which is defined as $\rho(t)\triangleq1+\nu\omega(t)^{T}\omega(t)$, where $\nu\in\mathbb{R}_{>0}$ is a gain. An update law is obtained using a gradient descent approach as $\dot{\hat{W}}_c(t) =-\Gamma(t)\frac{\partial E}{\partial \hat{W}_c}(t)$, which yields
\begin{equation}
\label{eq:critic_update}
\begin{aligned}
    \dot{\hat{W}}_c(t) =-\Gamma(t)\left(k_{c1}\frac{\omega(t)}{\rho^2(t)}\delta(t)+\frac{k_{c2}}{N}\sum_{k=1}^{N}\frac{\omega_{k}(t)}{\rho_{k}^2(t)}\delta_{k}(t)\right),
    \end{aligned}
\end{equation}
 where $\Gamma(t)\in\mathbb{R}^{L\times L}$ is a gain matrix that is updated according to 
\begin{equation}
    \label{eq:gamma_update}
    \dot{\Gamma}(t)=\beta\Gamma(t)-\Gamma(t)\left(k_{c1}\Lambda(t) +  \frac{k_{c2}}{N}\sum_{k=1}^{N}\Lambda_k(t) \right)\Gamma(t),
\end{equation}
where $\Lambda(t)\triangleq\frac{\omega(t)\omega(t)^{T}}{\rho^{2}(t)},\,\Lambda_k(t)\triangleq\frac{\omega_{k}(t)\omega_{k}(t)^{T}}{\rho_{k}^{2}(t)}$, and  $\beta\in\mathbb{R}_{>0}$ is a gain. Based on the analysis in Sec. \ref{sec:analysis}, the update law for $\hat{W}_a$ is selected as 

\begin{equation}
    \label{eq:actor_update}
    \dot{\hat{W}}_{a}(t)=\text{proj}\{ -k_{a1}(\hat{W}_a(t)-\hat{W}_c(t))\},
\end{equation}
where $k_{a1}\in\mathbb{R}_{>0}$ is a learning gain and $\text{proj}\{\cdot\}$ is a smooth operator\footnote{Details on the projection operator can be found in \cite{Dixon}.} which bounds the weight estimates. We make the following assumption to ensure weight estimate convergence:

\begin{assumption}[\cite{ADP_StaF}]\label{assumption:PE}
There exists constants $\underline{c}_1,\,\underline{c}_2,\,\underline{c}_3\in\mathbb{R}_{\geq0},\,T\in\mathbb{R}_{>0}$ such that 1) $\underline{c}_1I_L\leq\frac{1}{N}\sum_{k=1}^N\Lambda_k(t)$, 2) $\underline{c}_2I_L\leq\int_{t}^{t+T}\left(\frac{1}{N}\sum_{k=1}^N\Lambda_k(\tau) \right)d\tau,\,\forall t\in\mathbb{R}_{\geq0}$, 3) $\underline{c}_3I_L\int_{t}^{t+T}\left(\Lambda(\tau) \right)d\tau,\,\forall t\in\mathbb{R}_{\geq0}$ where at least one of $\underline{c}_i,\,i=1,2,3$ is strictly positive\footnote{The notation $I_L$ denotes an $L\times L$ identity matrix.}.
\end{assumption}
 If $\lambda_{\min}\{\Gamma^{-1}(0)\}>0$ and Assumption \ref{assumption:PE} holds, \eqref{eq:gamma_update} can be used to show that $\Gamma$ satisfies $\underline{\Gamma}I_{L}\leq\Gamma(t)\leq\overline{\Gamma}I_{L},\,\forall t\in\mathbb{R}_{\geq0}$ where $\underline{\Gamma},\,\overline{\Gamma}\in\mathbb{R}_{>0}$ and $\lambda_{\min}\{(\cdot)\}$ denotes the minimum eigenvalue of $(\cdot)$ \cite[Lemma 1]{ADP_StaF}.  

\section{ANALYSIS}\label{sec:analysis}
To aid in the analysis, we define the ideal weight estimate errors as $\tilde{W_{c}}\triangleq W-\hat{W}_{c}$ and $\tilde{W}_{a}\triangleq W-\hat{W}_{a}$. Now consider the Lyapunov function candidate
\begin{equation}
    \label{eq:Lyap}
    V_{L}(Z,t)\triangleq V^{*}+\frac{1}{2}\tilde{W}_{c}^{T}\Gamma^{-1}\tilde{W}_{c}+\frac{1}{2}\tilde{W}_{a}^{T}\tilde{W}_{a}
\end{equation}
and let $Z\triangleq[x^{T}\ \tilde{W}_{c}^{T}\ \tilde{W}_{a}^{T}]^{T}$. Note that the value function is positive definite, thus the Lyapunov function candidate is positive definite and can be bounded as $\eta_1(||Z||)\leq V_{L}(Z,t)\leq\eta_2(||Z||)$ where $\eta_1,\,\eta_2\,:\,\mathbb{R}_{\geq0}\rightarrow\mathbb{R}_{\geq0}$ are class $\mathcal{K}$ functions \cite[Lemma 4.3]{Khalil}. The sufficient conditions for the following theorem are $\psi k_{c2} >k_{a1},\,\eta^{-1}(\kappa) <\eta_{2}^{-1}(\eta_{1}(\xi)),\,x_0 \in \text{Int}(C)$ where $\xi\in\mathbb{R}_{>0}$ is the radius of the compact set used for value function approximation, $\psi\triangleq\left(\frac{\beta}{2k_{c2}\overline{\Gamma}}+\frac{\underline{c_1}}{2}\right)$, $\kappa\in\mathbb{R}_{>0}$ is a known positive constant that depends on the gains, and $\eta\,:\,\mathbb{R}_{\geq0}\rightarrow\mathbb{R}_{\geq0}$ is a class $\mathcal{K}$ function that satisfies 
\begin{equation*}
    \eta(\left\Vert Z\right\Vert)\leq
\frac{\underline{q}}{2}\left\Vert x\right\Vert^2+\frac{k_{a1}}{8}\left\Vert\tilde{W}_a\right\Vert^2+\frac{k_{c2}}{4}\psi\left\Vert\tilde{W}_c\right\Vert^2.
\end{equation*}

\begin{theorem}[Convergence and Safety] \label{thm:UUB}
  Given system \eqref{eq:dyn} under controller \eqref{eq:u_hat} with update laws \eqref{eq:critic_update}, \eqref{eq:gamma_update}, \eqref{eq:actor_update}, if Assumption \ref{assumption:PE} holds and the sufficient conditions are satisfied then $C$ is forward invariant and the state $x$ and weight estimation errors $\tilde{W}_c$, $\tilde{W}_a$ are uniformly ultimately bounded.
\end{theorem}

\begin{proof}
Omitted due to space constraints. Available upon request.
\end{proof}

\section{NUMERICAL EXAMPLES} \label{sec:sim}
In this section we present simulation results which were performed to assess the efficacy of our method and to compare it with the traditional QP approach. In the following, the system is simulated for 25 seconds under the influence of each controller. All differential equations are solved using Matlab's $\mathtt{ode45}$ function and \eqref{eq:CBF_QP_Sim} is solved using Matlab's $\mathtt{quadprog}$ function. Consider a two dimensional single integrator which can be represented as \eqref{eq:dyn} with $x\in\mathbb{R}^2$, $f=0_{2\times1}$, and $g=I_{2}$. The safe set is defined by \eqref{eq:safe_set} with 
  \begin{equation}
      \label{eq:h}
      h(x)\triangleq\sqrt{(x_{1}-z_{1})^{2}+(x_{2}-z_{2})^{2}}-r_h,
  \end{equation}
  where $z=\left[z_{1}\ z_{2}\right]^{T}$ denotes the center of the circular set, and $r_h\in\mathbb{R}_{>0}$ is its radius. For the approximate optimal controller we select the gains as $k_{c1}=0.05,\,k_{c2}=0.75,\,k_{a1}=0.75,\,\nu=1,\,\beta=0.001$. The cost function parameters are set to $Q=I_2,\,R=10I_{2}$ and the controller saturation is $\bar{u}=0.5$. The initial weights for the update laws are selected randomly from a uniform distribution between 0 and 4. The kernel function is defined by $\sigma(x,\,c(x))=\left[x^{T}c_{1}(x)\ x^{T}c_{2}(x)\ x^{T}c_{3}(x)\right]^{T}$ and we select the centers to be at the vertices of an equilateral triangle such that $c_{i}(x)=x+\vartheta(x)d_{i}$ where $\vartheta$ is a scaling factor defined as $\vartheta=\frac{0.5x^{T}x}{1+x^{T}x}$ and $d_{1}=\left[0\ -1\right]^{T}$, $d_{2}=\left[0.866\ -0.5\right]^{T}$, $d_{3}=\left[-0.866\ -0.5\right]^{T}$ are the center offsets.  To facilitate the finite excitation condition for weight convergence the BE is extrapolated to 1 random point from a $0.1\vartheta(x)\times0.1\vartheta(x)$ uniform distribution centered about $x$ at each time step. We select the barrier-like function as $B(x)=\frac{k_{p}s(x)}{h(x)}$ where $s$ is a smooth scheduling function and $k_p\in\mathbb{R}_{>0}$ is a gain. For the QP in \eqref{eq:CBF_QP_Sim} we define the CLF as $V_{clf}(x)=x^{T}Qx$. The functions $\alpha$, $\gamma$, and $p$ act as tuning parameters and are selected as $\alpha(h(x))=h(x)$, $\gamma(V_{clf}(x))=10V_{clf}(x)$, and $p=2$. To ensure results are comparable between methods we select $Q$, $R$, and $\bar{u}$ to be the same as in the ADP case. The results from applying each controller are shown in Fig. \ref{fig:state_space}-\ref{fig:run1Merged}. Fig. \ref{fig:state_space} illustrates each controller's ability to remain in $C$; however, the QP controller is incapable of converging to the origin. This behavior is further illustrated in Fig. ~\ref{fig:run1Merged} which is a result of introducing the relaxation variable $\varphi$ to ensure solvability of the QP. The gains on the CLF and relaxation variable can be tuned in an attempt to achieve better convergence; however, for no finite value of relaxation penalty $p$ can one ensure convergence to the equilibrium point \cite{JankovicAutomatica18}. Fig. \ref{fig:run1Merged} illustrates each controller's ability to satisfy the input and safety constraints.
  
  \begin{figure}
      \centering
      \includegraphics[clip, trim=0.5cm 6.5cm 0.5cm 6.5cm, width=0.48\textwidth]{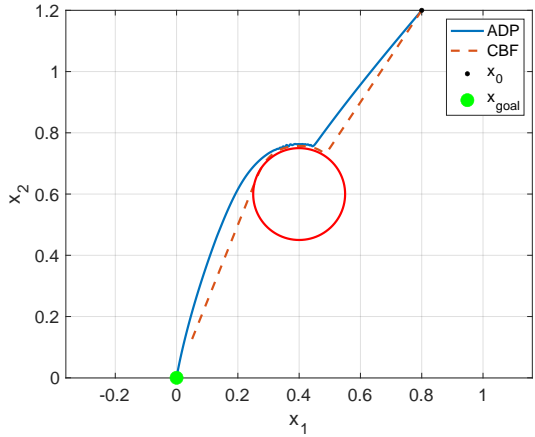}
      \caption{Trajectory of the system under ADP controller and QP controller. The boundary of $C$ is represented by the orange circle, and the origin is represented by a green dot.}
      \label{fig:state_space}
  \end{figure}

  \begin{figure}
      \centering
      \includegraphics[clip, trim=0.5cm 6.5cm 0.5cm 6.5cm, width=0.48\textwidth]{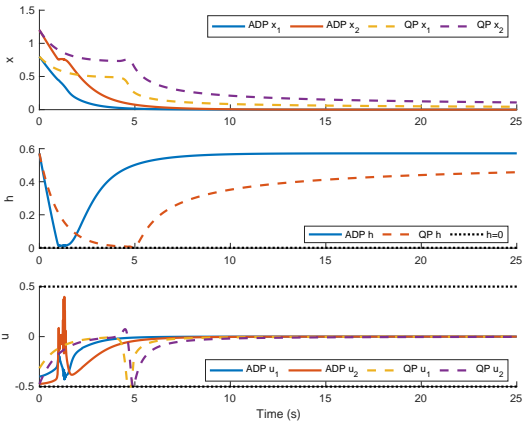}
      \caption{System states under each controller (top). Evolution of the barrier function under each controller (middle). Control trajectory over the course of the simulation (bottom).}
      \label{fig:run1Merged}
  \end{figure}

\section{CONCLUSION}\label{sec:conclusion}
We presented an alternative to the QP-based CBF approach to synthesizing optimal controllers for safety-critical systems. Instead, our method is based on ADP where we incorporate the cost of safety violation directly into the value function of an optimal control problem. We showed that the ADP method is able to guarantee both safety and stability of the resulting closed-loop system. We further illustrated this result with numerical examples in which the ADP controller outperformed the traditional QP controller in terms of convergence and feasibility. Future work will explore extending our approach to uncertain systems and differential games. 


\bibliographystyle{ieeetr}

  
\bibliography{ms}

\end{document}